\documentclass[11pt]{amsart}
\usepackage[margin=1.5in]{geometry} 
\usepackage{amsmath,amssymb}
\usepackage{enumerate}
\usepackage{xcolor}
\usepackage{graphicx}
\usepackage{dsfont}
\usepackage{tikz-cd} 
\usepackage[normalem]{ulem}

\newtheorem{theorem}{Theorem}[section]
\newtheorem{proposition}[theorem]{Proposition}
\newtheorem{lemma}[theorem]{Lemma}

\theoremstyle{definition}

\newtheorem{remark}[theorem]{Remark}

\theoremstyle{remark}
\newtheorem{step}{Step}

\newcommand{\N}{\mathbb{N}}

\newcommand{\R}{\mathbb{R}}

\newcommand{\cB}{\mathcal{B}}

\newcommand{\cF}{\mathcal{F}}

\newcommand{\cM}{\mathcal{M}}

\newcommand{\cQ}{\mathcal{Q}}

\newcommand{\cV}{\mathcal{V}}

\newcommand{\cMf}{\cM_\mathrm{fin}}
\newcommand{\tcMf}{\widetilde{\cM}_\mathrm{fin}}
\newcommand{\cVadm}{\cV_\mathrm{adm}}
\newcommand{\cVb}{\cV_\mathrm{b}}
\newcommand{\tQ}{\widetilde{Q}}

\newcommand{\ep}{\varepsilon}

\newcommand{\MNg}[1]{}%

\newcommand{\cP}{\mathcal{P}}

\newcommand{\as}{\mbox{-a.s.}}

\newcommand{\eps}{\varepsilon}

\newcommand{\1}{\mathbf{1}}
\DeclareMathOperator*{\argmin}{arg\, min}
\DeclareMathOperator*{\argmax}{arg\, max}

\numberwithin{equation}{section}

\let\originalleft\left
\let\originalright\right
\renewcommand{\left}{\mathopen{}\mathclose\bgroup\originalleft}
\renewcommand{\right}{\aftergroup\egroup\originalright}
\usepackage[pdfborder={0 0 0}]{hyperref}
\hypersetup{
  urlcolor = black,
  pdfauthor = {Marcel Nutz, Johannes Wiesel, Long Zhao},
  pdfkeywords = {martingale Schrodinger bridge, semistatic trading},
  pdftitle = {Martingale Schr\"odinger Bridges and Optimal Semistatic Portfolios},
  pdfsubject = {Martingale Schr\"odinger Bridges and Optimal Semistatic Portfolios},
  pdfpagemode = UseNone
}


\begin{document}
\title{Martingale Schr\"odinger Bridges and Optimal Semistatic Portfolios}
\date{\today}

\author{Marcel Nutz}
\thanks{MN acknowledges support by an Alfred P.\ Sloan Fellowship and NSF Grants DMS-1812661,  DMS-2106056.}
\address[MN]{Departments of Statistics and Mathematics, Columbia University,  1255 Amsterdam Avenue, New York, NY 10027, USA}
\email{mnutz@columbia.edu}

\author{Johannes Wiesel}
\address[JW]{Department of Statistics, Columbia University, 1255 Amsterdam Avenue,
New York, NY 10027, USA}
\email{johannes.wiesel@columbia.edu}

\author {Long Zhao }
\address[LZ]{Department of Statistics, Columbia University, 1255 Amsterdam Avenue,
New York, NY 10027, USA}
\email{long.zhao@columbia.edu}
\maketitle

\begin{abstract}
	In a two-period financial market where a stock is traded dynamically and European options at maturity are traded statically, we study the so-called martingale Schr\"odinger bridge~$Q_*$; that is, the minimal-entropy martingale measure among all models calibrated to option prices. This minimization is shown to be in duality with an exponential utility maximization over semistatic portfolios. Under a technical condition on the physical measure~$P$, we show that an optimal portfolio exists and provides an explicit solution for~$Q_{*}$. This result overcomes the remarkable issue of non-closedness of semistatic strategies discovered by Acciaio, Larsson and Schachermayer. Specifically, we exhibit a dense subset of calibrated martingale measures with particular properties to show that the portfolio in question has a well-defined and integrable option position.
\end{abstract}

\section{Introduction and Main Results}

The martingale Schr\"odinger bridge was introduced by~\cite{HenryLabordere.19} as a pricing model achieving perfect calibration to all Vanilla options while retaining stylized facts of a reference model. Starting from a reference stochastic volatility model (SVM) which typically cannot be calibrated perfectly, the martingale Schr\"odinger bridge is constructed as the calibrated measure which is closest to the SVM in the sense of relative entropy. In contrast to the classical Schr\"odinger bridge~\cite{Leonard.14} and~\cite{Avellaneda.98,AvellanedaEtAl.01}, this problem features an additional martingale constraint to generate an arbitrage-free model. 
A similar approach is used by~\cite{Guyon.20, Guyon.21} in a two-period setting to solve the longstanding  joint S\&P\,500/VIX smile calibration puzzle; here entropy minimization is utilized to construct a model that is jointly calibrated to the S\&P\,500, VIX futures and VIX options. %

The aforementioned works rest on (sometimes implicit) mathematical assumptions of strong duality and attainment. These are plausible as natural extensions of standard results in markets without option trading (see \cite{DelbaenEtAl.02,Frittelli.00,Schachermayer.01,Zariphopoulou.01}, among others). However, \cite{AcciaioLarssonSchachermayer.17} exhibited a surprising obstacle to obtaining such extensions: the space of semistatic  portfolios of stocks and options is not closed (both in a two-period model and in continuous time). In classical mathematical finance,  closedness results are at the very heart of the separation arguments underlying the Fundamental Theorem of Asset Pricing and the existence of optimal portfolios for utility maximization. As a consequence, it is not obvious how to formulate and prove the desired results. 

The purpose of the present paper is to provide such results, at least in one setting. On the one hand, we prove strong duality between the martingale Schr\"odinger bridge problem and an exponential utility maximization problem over semistatic portfolios. This duality, as well as the existence of the martingale Schr\"odinger bridge itself (primal attainment), is obtained along the lines of classical entropy minimization and Schr\"odinger bridge theory. On the other hand, we prove (under a technical condition) that the dual problem is attained in a natural space of admissible portfolios, and that this dual solution yields the log-density of the martingale Schr\"odinger bridge. We thus derive from first principles the type of implicit condition assumed on the optimal log-density, e.g., in \cite[Theorem~16]{Guyon.21}, and overcome the non-closedness issue discovered in~\cite{AcciaioLarssonSchachermayer.17}. To wit, while in general a convergent sequence of semistatic portfolios may have an undesirable limit with unclear financial interpretation, the specific limit of a utility-maximizing sequence in our problem is shown to be an admissible portfolio.

We consider a two-period model where the price of a stock is modeled by the canonical process~$(X,Y)$ on~$\R^{2}$ under a (physical) reference probability~$P$. Here~$X$ is the stock price at date $t=1$ and $Y$ is the price at the terminal date~$t=2$. In addition, European options $g(Y)$ are liquidly traded at time zero. By the Breeden--Litzenberger formula~\cite{BreedenLitzenberger.78}, the risk-neutral distribution~$\nu$ of~$Y$ can be derived from the prices of call options with arbitrary strikes, and then the arbitrage-free price of a general option~$g(Y)$ is given by the integral~$E^\nu[g]$. 
The martingale Schr\"odinger bridge problem can now be formalized as
\begin{equation} \label{eq: entropy min}
		\inf_{Q \in \cM(\nu)} H(Q|P),
\end{equation}
where $H$ is the relative entropy (or Kullback--Leibler divergence)
\begin{align*}
H(Q| P) := \begin{cases}
						E^Q\left[ \log \frac{dQ}{dP} \right], & Q \ll P \\
						\infty, & Q \not\ll P
					\end{cases}
\end{align*}
and~$\cM(\nu)$ is the set of calibrated equivalent martingale measures,
\begin{align}\label{eq: cMnu}
	\cM(\nu) := \left\{ Q\in \mathcal{P}(\R^{2}) :\ Q \sim P,\ Q^{2}= \nu, \ E^Q [Y |  X] = X \right\}.
\end{align}
Here $\mathcal{P}(\R^{2})$ is the set of probability measures on $\R^{2}$ and $Q^{2}$ denotes the second marginal of~$Q\in \mathcal{P}(\R^{2})$, or equivalently, the distribution of the price~$Y$ under~$Q$. 

We remark in passing that~\eqref{eq: entropy min} relates to the classical (static) Schr\"odinger bridge  problem $\inf_{Q \in \Pi(\mu,\nu)} H(Q|P)$ over the set~$\Pi(\mu,\nu)$ of couplings of two measures~$\mu,\nu$; see \cite{Follmer.88,Leonard.14,Nutz.20} for surveys. In this problem, there is no martingale constraint. On the other hand, \eqref{eq: entropy min} relates to the martingale optimal transport problem $\inf_{Q \in \mathsf{M}(\mu,\nu)} E^{Q}[c]$ which minimizes an integrated cost over the set $\mathsf{M}(\mu,\nu)$ of martingale couplings; see \cite{BeiglbockHenryLaborderePenkner.11,GalichonHenryLabordereTouzi.11,Hobson.98} and the literature thereafter. In that problem, there is no reference measure. The resulting value yields model-independent bounds for the price of the exotic option~$c$ and as a consequence of the linear structure, solutions tend to be degenerate. By contrast, solutions of~\eqref{eq: entropy min} tend to preserve features of the reference model~$P$, as emphasized in~\cite{HenryLabordere.19}.
The classical Schr\"odinger bridge  problem arises from the classical optimal transport problem by entropic regularization as used in the context of Sinkhorn's algorithm~\cite{Cuturi.13, PeyreCuturi.19}. Similarly, entropic regularization of martingale optimal transport leads to the martingale Schr\"odinger bridge, and this was used in~\cite{DeMarchHenryLabordere.19} to develop a version of Sinkhorn's algorithm for martingale optimal transport. See also~\cite{GuoObloj.19} for a related algorithm using a different relaxation.

Returning to our problem~\eqref{eq: entropy min}---for it to be meaningful, we must assume that 
\begin{align}\label{eq: cMf}
\cMf( \nu) := \{ Q \in \cM( \nu): H(Q |  P ) < \infty \} \neq \emptyset;
\end{align}
that is, there exists a calibrated martingale measure with finite entropy. This condition implies the absence of arbitrage in semistatic trading strategies. It implies the usual no-arbitrage condition on the stock alone, but also depends on but also depends on the interplay of~$P$ and~$\nu$. A precise characterization of~\eqref{eq: cMf}, or even just $\cM( \nu)\neq\emptyset$, in terms of trading strategies along the lines of a fundamental theorem of asset pricing~\cite{DalangMortonWillinger.90}, is an interesting open problem. (Like the question studied in the present paper, the answer is not obvious due to the failure of closedness~\cite{AcciaioLarssonSchachermayer.17}.)
We can now state the basic wellposedness result.

\begin{proposition}\label{prop: existence}
The problem~\eqref{eq: entropy min} admits a unique minimizer $Q_{*}\in\cM(\nu)$, called the martingale Schr\"odinger bridge.
\end{proposition}

This will essentially follow from standard entropy minimization theory~\cite{Csiszar.75} and properties of $\mathcal{M}(\nu)$ which are variations of results found, e.g., in~\cite{BeiglbockHenryLaborderePenkner.11}.
Proposition~\ref{prop: existence} lacks a more specific description: we expect by (formal) duality that the log-density of~$Q_*$ corresponds to a semistatic portfolio with certain admissibility criteria, and those criteria are crucial for any further analysis of the martingale Schr\"odinger bridge and its computation (as seen, e.g., in~\cite{Guyon.21}). 
Specifically, trading in our market gives rise to a semistatic outcome of the form
$$
  V=h(X)(Y-X) + g(Y),
$$
where $h(X)$ is the number of stocks held over the second period. Stock trading in the first period, starting from a deterministic initial stock price~$X_{0}$, corresponds to a term $h_{0}(X_{0})(X-X_{0})$ which can be absorbed into the functions $h,g$ above and hence will not be represented explicitly.
We write
	\begin{equation} \label{eq: semistatic}
		 \cV= \{V \ \mbox{measurable:}\ V=h(X) (Y-X) + g(Y) \mbox{ for some } h,g:\R\to \R \}.
	\end{equation}
  In order to have a well-defined option price, the function~$g$ needs to be (measurable and) integrable under the pricing measure~$\nu$. We thus set
\begin{equation} \label{eq: V1}
  \cV_{1}:= \{V\in\cV: \ h,g \mbox{ are measurable,}\ g\in L^{1}(\nu), \ E^{\nu}[g]=0 \}
\end{equation}
for those outcomes whose option is available from zero initial capital. 
Finally, we want $h(X) (Y-X)$ to have suitable martingale properties. There is some flexibility here regarding the definition; one natural choice is to require the martingale property under all $Q\in \cMf(\nu)$ (see also Remark~\ref{rk:differentAdm} for another possible choice).
For $V\in \cV_{1}$, this is equivalent to~$V\in L^1(Q)$ for all $Q\in\cMf(\nu)$.
In summary, our set of admissible portfolios (for zero initial capital) is
$$
  \cVadm= \left\{V\in\cV: 
  \begin{array}{c}
  h,g:\R\to \R  \mbox{ are measurable}, \
  E^{\nu}[g]=0, \\
  E^{Q}[h(X)(Y-X)]=0 \mbox{ for all }Q \in \cMf( \nu)
  \end{array}
  \right\}.
$$
We then have the following strong duality between the martingale Schr\"odinger bridge (primal) problem and the dual problem of exponential utility maximization over semistatic portfolios.

\begin{proposition}\label{pr:duality}
Let $u(x)=-e^{-\gamma x}/\gamma$ for some $\gamma>0$. Then
 \begin{equation} \label{eq: duality}
		\frac{1}{\gamma}  \inf_{Q \in \cM( \nu)} H(Q| P)= \sup_{V \in \cVadm} u^{-1}\left( E^P[ u(V) ] \right).
	\end{equation}
\end{proposition}

The duality will be obtained by showing that the log-density of $Q_{*}$ can be approximated by semistatic portfolios with good integrability properties; cf.\ Proposition~\ref{prop: characterisation}. That proposition, in turn, is inspired by seminal results in the theory of (classical) Schr\"odinger bridges, especially F\"ollmer's construction of Schr\"odinger potentials~\cite{Follmer.88}. 
Our argument does not require dual attainment and thus avoids discussing delicate properties of the portfolios: the supremum in~\eqref{eq: duality} would be the same if taken, say, over portfolios $h(X) (Y-X) + g(Y)$ with bounded continuous functions~$h,g$. But of course, this space would not allow for attainment in general.

Turning to the delicate part, we want to show that the dual problem is attained at an admissible portfolio~$V_{*}$ and that this maximizer yields the log-density of $\cQ_{*}$. We denote by $P=P^{1}\otimes P^{\bullet}$ the disintegration of~$P$; that is, $P^{1}$ is the law of $X$ under~$P$ and~$P^\bullet(x,dy)$ is the conditional law of~$Y$ given~$X=x$.

\begin{theorem}\label{thm: integrability}
  Suppose that $dP^\bullet/d\nu$ is $P^1$-a.s.\ uniformly bounded from above and below. 
  Then the minimizer~$Q_{*}$ of~\eqref{eq: entropy min} is given by the density
  \begin{align}\label{eq: density std form}
  Z_*:=\frac{dQ_*}{dP}=e^{H(Q_*| P) + V_*},
  \end{align}	
  where $V_{*}\in\cVadm$ is the unique solution of the dual problem,
  $$
    V_{*} = \argmax_{V\in\cVadm} E^P[ u(V) ] .
  $$
  In particular, $V_{*}=h(X)(Y-X)+g(Y)$, where $h,g$ are measurable functions with $g\in L^{1}(\nu)$ and $E^{\nu}[g]=0$ as well as $h(X)(Y-X)\in L^{1}(Q)$ and $E^{Q}[h(X)(Y-X)]=0$ for all $Q \in \cMf(\nu)$.
\end{theorem}

The boundedness condition in Theorem~\ref{thm: integrability} can be weakened to an integrability condition; see Remark~\ref{rk:relaxBddnessCond}. 
In contrast to the other results, this theorem does not seem to follow from classical arguments. If the space of admissible portfolios were closed, the theorem would follow from the approximation result in Proposition~\ref{prop: characterisation}, broadly as in the classical framework of mathematical finance without options. To overcome the failure of closedness (specifically, of $\cV_{1}$ and $\cVadm$, as shown in~\cite{AcciaioLarssonSchachermayer.17}), we first leverage a result from our companion paper~\cite{NutzWieselZhao.22a}, where it is shown that the functional form of semistatic portfolios is stable under pointwise limits. As a consequence, the approximation result still implies that~$V_{*}$ is of the general form $V_{*}=h(X)(Y-X)+g(Y)$ for some measurable functions~$h,g$.

On the flip side, another insight from~\cite{NutzWieselZhao.22a} is that the key failure in the counterexample of~\cite{AcciaioLarssonSchachermayer.17} is the integrability of the option~$g$  which is in turn crucial to associate a price. Hence, it is not surprising that establishing this integrability occupies the lion's share  of the proof of Theorem~\ref{thm: integrability}; it uses novel arguments and seems to be the first result in this direction. Our line of attack is to construct a measure $\widetilde{Q}$ in (a relaxation of) $\cMf(\nu)$ such that $h(X)(Y-X)$ is $\widetilde{Q}$-integrable; once that is achieved, soft arguments imply that~$g\in L^1(\nu)$. In fact, we establish that such measures are dense: in Proposition~\ref{prop: integrability approx} we show that any $Q\in \cMf(\nu)$ is the limit of calibrated (absolutely continuous) martingale measures~$Q_n$ under which the dynamic trading strategy~$h$
is uniformly bounded a.s.  The proof is intricate and develops, among other things,  explicit stability properties of the convex order, building on ideas from martingale optimal transport~\cite{BeiglbockJuillet.12}. See also Section~\ref{sec: integrability approx} for further comments.

We do not know how far the technical condition on~$P$ in Theorem~\ref{thm: integrability} can be relaxed. However, analogy with the classical Schr\"odinger bridge problem suggests that some condition may be necessary. Indeed, the corresponding question in that setting---without martingale constraint but with two marginal constraints---is to show that the log-density of the Schr\"odinger bridge is of the form $f(x)+g(y)$ and establish the measurability and integrability properties of those ``Schr\"odinger potentials''~$(f,g)$. This problem has a long history (e.g., \cite{Beurling.60}). A series of results revealed that the additive form $f(x)+g(y)$ always holds, but also that the measurability of~$(f,g)$ fails without additional conditions; moreover, even when measurability holds, integrability fails without further conditions (see~\cite{BorweinLewis.92, Csiszar.75, FollmerGantert.97, RuschendorfThomsen.93, RuschendorfThomsen.97}).
The study of Schr\"odinger potentials remains an area of active study (see for instance~\cite{AltschulerNilesWeedStromme.21, DeligiannidisDeBortoliDoucet.21, GigliTamanini.21, NutzWiesel.21, NutzWiesel.22}) that we have benefited from, especially for our companion paper~\cite{NutzWieselZhao.22a}. For the present work, we have not been able to transfer as many of those techniques. 

Regarding potential future work, it seems likely that our line of argument can be extended to show the existence of optimal portfolios for more general utility functions. Generalizations in the structure of the market, for instance also adding options with maturity $t=1$, are relatively straightforward in the general parts whereas replacing our argument for the integrability of the option is nontrivial. In a different direction, one may remember how~\cite{Rogers.94} used existence for exponential utility to show the Fundamental Theorem of Asset Pricing. Of course, that is not immediately applicable here, as we have used~\eqref{eq: cMf} in our proof of existence.

The remainder of the paper has a simple structure: Section~\ref{se:wellposedDuality} derives the wellposedness and duality results (Propositions~\ref{prop: existence} and~\ref{pr:duality}), and Section~\ref{se:dualAtt} provides the proof of dual attainment (Theorem~\ref{thm: integrability}).

\section{Wellposedness and Duality}\label{se:wellposedDuality} %

In this section, we first prove the wellposedness of the martingale Schr\"odinger bridge~$Q_{*}$ (Proposition~\ref{prop: existence}). Then, we prove the duality with exponential utility maximization (Proposition~\ref{pr:duality}) through an approximation of~$Q_{*}$ (Proposition~\ref{prop: characterisation}).

We start by recalling a general result on entropy minimization.

\begin{lemma} \label{lemma: min entropy optimizer}
	Consider a measurable space $(\Omega, \cF)$ and denote by $\cP(\Omega)$ its collection of probability measures. Fix $R \in \cP(\Omega)$, let $ \cQ \subseteq \cP(\Omega)$ be convex and closed in variation, and suppose that $\cQ_{\mathrm{fin}} := \{ Q \in \cQ : H(Q|R) < \infty \} \neq \emptyset.$ Then there exists a unique $Q_* \in \cQ$ such that
		$$
			H(Q_*|R) = \inf_{Q \in \cQ} H(Q|R) \in [0, \infty).
		$$
Moreover, $Q_* \gg Q$ for any $Q \in \cQ_{\mathrm{fin}}$. In particular, if there exists $Q \in \cQ_{\mathrm{fin}}$ with $Q \sim R$, then $Q_* \sim R$. Furthermore, 
\begin{equation}\label{eq:entropyMinIntegrab}
  \log \frac{dQ_*}{dR} \in L^1(Q) \quad \mbox{for all} \quad Q \in \cQ_{\mathrm{fin}}.
\end{equation}
\end{lemma}
\begin{proof}
	In the stated form, the result can be found in \cite[Theorem~1.10 and Corollary~1.13]{Nutz.20}. Its main part is very classical; cf.~\cite{Csiszar.75}. The integrability~\eqref{eq:entropyMinIntegrab} is less known but can also be deduced from~\cite{Csiszar.75}.
\end{proof}

Lemma \ref{lemma: min entropy optimizer} is not directly applicable to the set $\cQ = \cM(\nu)$ of martingale measures defined in~\eqref{eq: cMnu} as this set is not closed due to the equivalence constraint. 
Writing $$\Pi(\nu)=\{Q\in\cP(\R^{2}):\, Q^{2}=\nu\},$$ we consider instead the following relaxations defined with absolute continuity,
\begin{align}
	\widetilde{\cM}( \nu) &:= \left\{ Q \in \Pi( \nu): \ Q \ll P, \ E^{Q} [Y| X] = X \right\} \supseteq \cM(\nu), \nonumber \\
	\widetilde{\cM}_{\mathrm{fin}}(\nu) &:= \{ Q \in \widetilde{\cM}( \nu): H(Q |  P ) < \infty \} \supseteq \cMf(\nu)  \label{eq: tilde cMf}
\end{align} 
and argue that $\widetilde{\cM}(\nu)$ satisfies the hypotheses of Lemma~\ref{lemma: min entropy optimizer}. To this end, we first give an extension of \cite[Lemma~2.2 and Theorem~2.4]{BeiglbockHenryLaborderePenkner.11}.
Recall that two measures $\mu,\nu\in\cP(\R)$ are in convex order~\cite{ShakedShanthikumar.07}, denoted $\mu\preceq_c \nu$, if they have finite first moments and
$
E^\mu[f]\le E^\nu[f] 
$
holds for all convex functions $f:\R\to \R$. As an example, $E^{Q} [Y| X] = X$ implies $Q^{1}\preceq_c Q^{2}$ by Jensen's inequality.

\begin{lemma}\label{lem:compact}
The set
$
 \{ Q \in \Pi( \nu): \ E^{Q} [Y| X] = X  \}
$
is weakly closed.
\end{lemma}

\begin{proof}
  Let $(Q_n)_{n\geq1}$ be a sequence of measures converging  weakly to some limit~$Q$, then $Q\in\Pi(\nu)$ by the continuity of the projection~$Y$. To see that $E^{Q_{n}} [Y| X] = X$ implies $E^{Q} [Y| X] = X$, we show that $| X| +| Y|$ is $(Q_n)$-uniformly integrable.
Indeed, as $\{\nu\}$ is uniformly integrable,
the la Valle\'e--Poussin theorem yields a convex function $f:\R\to\R_{+}$ of superlinear growth with $\int f\, d\nu<\infty$. Thus
$$
  \sup_{\mu:\,\mu\preceq_c \nu} \int f\, d\mu \leq \int f\, d\nu<\infty
$$
by the definition of the convex order, showing that $\{\mu:\,\mu\preceq_c \nu\}$ is uniformly integrable.  As a result, $| X| +| Y|$ is $\{ Q\in \Pi(\nu): Q^1 \preceq_c \nu\}$-uniformly integrable and in particular $(Q_n)$-uniformly integrable.
\end{proof}

We can now show the wellposedness of the martingale Schr\"odinger bridge~$Q_{*}$.

\begin{proof}[Proof of Proposition \ref{prop: existence}]	
Using Lemma~\ref{lem:compact}, we readily verify that $\widetilde{\cM}(\nu)$ is convex and closed in variation. 
Since $\widetilde{\cM}_{\mathrm{fin}}(\nu) \supseteq \cMf(\nu)\neq\emptyset$ by our assumption~\eqref{eq: cMf}, applying Lemma \ref{lemma: min entropy optimizer} with $\mathcal{Q} = \widetilde{\cM}( \nu)$ yields existence and uniqueness of
\begin{align}\label{eq:relaxedArgmin}
 Q_{*}=\argmin_{Q\in \widetilde{\cM}( \nu)} H(Q| P)
\end{align} 
as well as $Q_* \sim P$; that is, $Q_*\in \cMf(\nu)$. It now follows that $Q_*$ is also the unique minimizer of $\inf_{Q \in \cM(\nu)} H(Q|P)$.
\end{proof}

We record the following observation for use in Section~\ref{se:dualAtt}.

\begin{remark}\label{rmk: psi star int}
For any $Q \in \cMf(\nu)$, a straightforward calculation shows that the density $Z := dQ/dP$ can be written as $Z = e^{H(Q| P) + V}$ for some $V \in L^1(Q)$ with $E^Q[V] = 0$. For the density 
\begin{align}\label{eq: density std form}
  Z_*:=\frac{dQ_*}{dP}=e^{H(Q_*| P) + V_*},
\end{align}	
of the minimizer, we have not only that $E^{Q_{*}}[V_{*}] = 0$ but also that $V_* \in L^1(\widetilde Q)$ for all $\widetilde Q \in \widetilde{\cM}_{\mathrm{fin}}(\nu)$. This follows from~\eqref{eq:entropyMinIntegrab} by way of~\eqref{eq:relaxedArgmin}.
\end{remark}

The next result characterizes the minimizer~$Q_{*}$ through certain approximating sequences of semistatic portfolios and will serve as the basis to prove the duality (Proposition~\ref{pr:duality}). We write $\cVb$ for the set of portfolios $V=h(X) (Y-X)  + g(Y)$ where $h,g: \R\to\R$ are bounded measurable and $E^\nu[ g ] = 0$; clearly $\cVb\subset\cVadm$.

\begin{proposition}\label{prop: characterisation}
	Given $Q_*\in\cMf(\nu)$ with density~\eqref{eq: density std form}, the following statements are equivalent:
	\begin{enumerate}[(i)]
	\item $Q_*$ is the minimizer of \eqref{eq: entropy min}.
	\item There exist probability measures $(Q_n)_{n\geq 1}$ with densities $$Z_n := dQ_n/dP=e^{H(Q_n| P)+ V_n} \quad\mbox{with} \quad V_n \in \cVb$$
	such that 
	$$
	H(Q_n| P)\to H(Q_*| P) \quad \text{and}\quad V_n  \to V_* \text{ in } L^1(Q_*) \quad \text{as } n \to \infty.
	$$
	
	\item There exist $(V_n)_{n\geq 1} \subseteq \cVb$ such that%
	$$E^P \big[ e^{V_n} \big] \to E^P \big[ e^{V_*}\big] \quad \text{as } n \to \infty.$$
	\item There exist $(V_n)_{n\geq 1} \subseteq \cVb$ such that
	$$E^{Q} \big| e^{V_n-V_*} - 1\big|  \to 0 \quad \text{as } n \to \infty.$$
	\end{enumerate}
\end{proposition}

This result is inspired by a characterization of (classical) Schr\"odinger bridges in \cite[Proposition~3.6]{FollmerGantert.97}, see also~\cite{Csiszar.75}, and F\"ollmer's construction of Schr\"odinger potentials~\cite{Follmer.88}. The key feature is that the approximating random variables~$V_{n}$ are portfolios and have good integrability properties (whereas the properties of $V_{*}$ are unclear at this stage).

\begin{remark}
The assertion of Proposition~\ref{prop: characterisation} remains valid if~$\cVb$ is replaced by~$\cVadm$ or by~$\cV_{1}$. This will be clear from the proof.
\end{remark}

\begin{proof} [Proof of Proposition \ref{prop: characterisation}]
	$(i) \Rightarrow (ii)$: By separability of $L^1(\R)$ we can write
		\begin{align*}
		\cM( \nu) = \{ Q \sim P : & \ E^Q[  h_i(X) (Y-X) ] = 0, \ E^Q[g_i(Y)] = 0, \ i = 1,2,\dots \}
		\end{align*}
		for a countable collection of bounded measurable functions $h_i,g_i : \R \to \R$. Denote the set of measures for which only the first $n$ constraints are enforced by 
		\begin{align*}
		\cM_n(\nu)= \{ Q \sim P : & \ E^Q[  h_i(X) (Y-X) ] = 0, \ E^Q[g_i(Y)] = 0, \ i = 1,2,\dots,n \}.
		\end{align*}				
		Clearly $\cM(\nu)\subseteq \cM_n(\nu)$ and $\cM_n(\nu)$ is convex and closed in variation. Consider the problem $\inf_{Q\in \cM_n(\nu)} H(Q| P)$. This minimization problem over measures with finitely many linear constraints is well known to be in duality with exponential utility maximization over (static) trading in the finitely many assets $h_i(X) (Y-X), g_i(Y)$ defining the constraints. Specifically, by \cite[Section~3, esp.\ Corollary~3.25]{FollmerSchied.11}, the minimizer $Q_n$ 
		of $\inf_{Q\in \cM_n(\nu)} H(Q| P)$ is of the form
		$$
			Z_n := \frac{dQ_n}{dP} = \exp \left( c_n + \tilde{h}_n(X) (Y-X) + \tilde{g}_n(Y) \right)
		$$
		for some $c_n \in \R$, where $\tilde{h}_n(X) = \sum_{i=1}^{n} a_{i,n} h_{i}(X)$ and $\tilde{g}_n(Y) =\sum_{i=1}^{n} b_{i,n}g_i(Y)$ for some  $a_{i,n}, b_{i,n}, 				\in \R$. As $Q_n\in \cM_n(\nu)$ we have  $c_n = H(Q_n| P)$. That is, $\log Z_n$ is of 
		the form $H(Q_n| P) + V_n$ for some $V_n \in \cVb$. Applying \cite[Theorem 1.17]{Nutz.20} to the sets $\cQ_n := \cM_n( \nu)$ and $\cQ:=\cM(\nu)$ satisfying $\cap_{n}\cQ_n=\cQ$, and
		recalling 
		that $\cM_\mathrm{fin}(\nu)\neq \emptyset$ by our assumption~\eqref{eq: cMf}, we conclude that
		\begin{align*}
		H(Q_*| Q_n)\to 0, \quad H(Q_n| P)\to H(Q_*| P) \quad \text{and}\quad \log Z_n \to \log Z_* \text{  in }L^1(Q_*).
		\end{align*}
		In particular, $V_n \to V_*$ in $L^1(Q_*)$ follows.\\
	$(ii) \Rightarrow (iii)$: Since $Z_n, Z$ are probability densities, we have  $e^{-H(Q_n| P)}= E^P[e^{V_n}]$ and $e^{-H(Q_*| P)}=E^P[e^{V_*}]$. Thus $H(Q_n| P) \to H(Q_*| P)$ is equivalent to $E^P[ e^{V_n}] \to E^P[ e^{V_*}]$. \\
	$(iii) \Leftrightarrow (iv):$ By a change of measure, (iii) is equivalent to
	\begin{align*}
	E^{Q_*} [e^{V_n - V_*}] = e^{H(Q_*| P)} E^P[e^{V_n}]\to e^{H(Q_*| P)}  E^P[e^{V_*}]=1,
	\end{align*}
	and now Scheff\'e's lemma yields the equivalence with~(iv).\\
	$(iii) \Rightarrow (i)$: Without loss of generality, we assume $E^P [ e^{V_n} ] < \infty$ for all $n$. Define probability measures $Q_n$ by
		$$
		Z_n := \frac{dQ_n}{dP} =  e^{H(Q_n| P) + V_n}
		$$
		and recall that (iii) is equivalent to $H(Q_n| P)\to H(Q_*| P)$.
		Take any $Q \in \cMf(\nu)$. Using the definition of $H(\cdot | P)$ and Lemma~\ref{lemma: integrability} below,
		$$
		H(Q| P) - H(Q| Q_n) = E^{Q}\left[ \log Z_n \right] = H(Q_n| P) + E^{Q} [V_n] = H(Q_n| P).
		$$
		As $H(Q| Q_n) \geq 0$, it follows that
		$$
		H(Q| P) \geq \lim_{n \to \infty} H(Q_n| P) = H(Q_*| P).
		$$
		Since $Q \in \cMf(\nu)$ was arbitrary, we conclude that $Q_*$ is the minimizer of \eqref{eq: entropy min}.
\end{proof}

The following technical result was used in the preceding proof.

\begin{lemma} \label{lemma: integrability}
	Let $V\in \cV_1$ satisfy $E^P[e^V]<\infty$. Then $V \in L^1(Q)$ and $E^{Q}[V] = 0$ for all $Q \in \cM_\mathrm{fin}(\nu)$. 
\end{lemma}
\begin{proof}
	Define an auxiliary probability measure $Q'$ via
	\begin{align*}
	Z':= \frac{dQ'}{dP} =  e^{c + V},
	\end{align*}	
	where $c\in \R$ is the normalization constant. Moreover, let $Q \in \cM_\mathrm{fin}(\nu)$ and denote by~$Z$ its density.  
 Applying the inequality $\log x \leq x -1$ to $x = z'/z > 0$ yields $\log z' \leq \log z + z'/z - 1$ and hence
	$$
		\log Z' \leq \log Z +Z'/Z - 1 \quad \text{on} \quad \{ Z > 0 \},
	$$
	where $\log 0 := - \infty$.  In view of $H(Q| P) < \infty$, we have $\log Z + Z'/ Z - 1 \in L^1(Q)$ and conclude that $(\log Z')^{+} \in L^{1}(Q)$.
  By the definition of $\cV_1$,
  $$\log Z' = c + V = h(X)(Y-X)+ g(Y)$$
  for some $g \in L^1(\nu)$ with $E^\nu[ g] = 0$, and hence $(\log Z')^{+} \in L^{1}(Q)$ translates to the positive part of the martingale transform $h(X) (Y-X)$ being $Q$-integrable. As~$Q$ is a martingale measure, this already implies (see \cite[Theorem~2b]{JacodShiryaev.98}) that $h(X)(Y-X) \in L^{1}(Q)$ and $E^{Q}[h(X)(Y-X)]=0$. The claim follows.
\end{proof}

We are now in a position to prove the duality result.

\begin{proof}[Proof of Proposition~\ref{pr:duality}]
Let $Q_*$ be the minimizer from Proposition~\ref{prop: existence} and recall from \eqref{eq: density std form} the notation $Z_* = dQ_*/dP = e^{ H(Q_*| P) + V_*}$ where $E^{Q_*}[ V_*]=0$. Let $u(x)=-e^{-\gamma x}/\gamma$ for some $\gamma>0$.
  A change of measure and Jensen's inequality yield that for any $V\in \cVadm$,
	\begin{align*}
	    E^P [u(V)] &= E^{Q_*} \left[ Z_*^{-1} u (V) \right]  = -\frac{1}{\gamma}  E^{Q_*} \left[ e^{-H(Q_*| P) -V_* - \gamma V} \right] \\
		& \leq -   \frac{1}{\gamma} e^{ -H(Q_*| P) -E^{Q_*}[ V_*] - \gamma E^{Q_*}[V] }=  - \frac{1}{\gamma}  e^{ -H(Q_*| P)},
	\end{align*}
	where the equality used that $E^{Q_*}[V]= 0$ due to $Q_*\in \cMf(\nu)$ and $V\in \cVadm$.

	On the other hand, Proposition \ref{prop: characterisation}\,(iv) shows that there exist $( V_n) \subseteq \cVb$ such that $E^{Q_*} \left[  e^{ - V_* -\gamma V_n } \right] \to 1$ and consequently $-  E^{Q_*} \left[  e^{ -H(Q_*| P) -V_* -\gamma V_n} \right] \to -e^{-H(Q_*| P)}$.
	In view of $\cVb\subseteq \cVadm$, this yields $\sup_{V \in \cVadm} E^{P}[u(V)] \geq -  \frac{1}{\gamma}e^{ H(Q_*| P)}$.
	Lastly,
	$$
		 \inf_{Q \in \cM( \nu)} u\left(\frac{1}{\gamma} H(Q| P)\right) = u\left(\frac{1}{\gamma}H(Q_*| P)\right)  = - \frac{1}{\gamma} e^{ - H(Q_*| P)},
	$$
	so that combining the two inequalities yields
	\begin{align*}
	\sup_{V \in \cVadm} E^{P}[u(V)] =\inf_{Q \in \cM( \nu)} u\left(\frac{1}{\gamma} H(Q| P)\right)
	\end{align*}
	as claimed.
	\end{proof}

\begin{remark}\label{rk:differentAdm}
  By the proof, the duality~\eqref{eq: duality} still holds if the supremum is taken over the larger set $\cV_1 \cap L^1(Q_*) \supset \cVadm$, providing an alternative definition of admissibility.
\end{remark}

\section{Admissibility and Dual Attainment}\label{se:dualAtt}

\subsection{Preliminary Considerations}\label{sec: dualAttPrelim}

Let $Q_*$ be the minimizer from Proposition~\ref{prop: existence} and recall from~\eqref{eq: density std form} the notation $Z_* = dQ_*/dP = e^{ H(Q_*| P) + V_*}$. With a view towards the duality relation, note that
	$$
		E^P\left[u\left(-\frac{1}{\gamma}V_*\right)\right] = - E^P\left[ \frac{1}{\gamma} e^{ V_*} \right] = - E^P\left[ \frac{1}{\gamma} e^{ -H(Q_*| P)} Z_*\right] = -\frac{1}{\gamma}  e^{ -H(Q_*| P)}.
	$$
	It is thus tempting to conclude that $-V_*/\gamma$ ``attains" the supremum in \eqref{eq: duality}. However, it far from obvious whether $V_*$ belongs to the dual domain $\cVadm$ (or is a portfolio in any sense). At this stage, we know that $E^{Q_*}[ V_*]=0$ and that $V_{*}$ is the limit of certain portfolios $V_{n}\in \cVb \subset \cVadm\subset \cV_{1}$; cf.~Proposition~\ref{prop: characterisation}. The missing conclusion would be obvious if any of these spaces had a good closure property. However, as mentioned in the Introduction, \cite{AcciaioLarssonSchachermayer.17} has shown that this is not the case: specifically, the authors exhibit a two-period model and an $L^{p}$-convergent sequence $V_{n}\in \cVb$ whose limit is outside~$\cV_{1}$. The proof uses a clever contradiction argument avoiding a detailed study of the limiting random variable, and so it may not be clear what exactly goes wrong in the limit.
	
The first possible issue is whether the limit still has the functional form $h(X)(Y-X) + g(Y)$ for some functions $h,g$. A second issue is whether (these functions are measurable and) $g$~is integrable as required by the definition of~$\cV_{1}$. The first issue is analyzed in our companion paper~\cite{NutzWieselZhao.22a} which shows that the functional form is stable even under pointwise limits. Under the mild condition that~$P\sim P^{1}\otimes P^{2}$  (which is implied by the condition in Theorem~\ref{thm: integrability}), we can also guarantee that $h,g$ remain measurable.

\begin{lemma}\label{le:functionalClosedness}
  We have $V_{*}\in\cV$; that is, $V_{*}=h(X)(Y-X) + g(Y)$ for some functions $h,g:\R\to\R$. If $P\sim P^{1}\otimes P^{2}$, the functions $h,g$ are a.s.\ uniquely determined and measurable.
\end{lemma} 

\begin{proof}
  By Proposition~\ref{prop: characterisation} we can find $V_{n}\in\cVb$ with $V_{n}\to V_{*}$ $P$-a.s. The two claims then follow from~\cite[Theorem~2.2]{NutzWieselZhao.22a} and~\cite[Theorem~3.1]{NutzWieselZhao.22a}, respectively. 
\end{proof} 

This stability of the functional form indicates that the key failure in the counterexample of~\cite{AcciaioLarssonSchachermayer.17} is the integrability of the option. It is then clear that some original arguments will be required to obtain that the option position in our specific limit~$V_{*}$ is nevertheless integrable---which motivates the rest of this section.

\subsection{Proof of Theorem~\ref{thm: integrability}}\label{sec: dualAttProof}

Recall that the disintegration of a probability measure~$R \in \cP(\R^2)$ is denoted $R=R^{1}\otimes R^\bullet$, where $R^{1}$ is the first marginal (distribution of $X$) and $R^\bullet: \R \to \cP(\R)$ is a stochastic kernel (conditional distribution of~$Y$ given~$X$). We interchangeably use $R^\bullet(x)$, $R^\bullet(x,\cdot)$ or $R^\bullet(x,dy)$ to denote the conditional distribution given~$X=x$.

Our basic line of attack is simple (yet seems to be novel): recalling Remark~\ref{rmk: psi star int}, 
\begin{equation}\label{eq:entropyFinIntegrab}
  V_{*}=h(X)(Y-X)+ g(Y) \in L^1(Q) \quad \mbox{for all} \quad Q \in \tcMf(\nu),
\end{equation}
where $\tcMf(\nu)$ was defined in~\eqref{eq: tilde cMf}. We shall construct $\tQ \in \tcMf(\nu)$ such that $h$ is uniformly bounded $\tQ^1$-a.s. Then clearly $h(X)(Y-X)\in L^{1}(\tQ)$ and now ~\eqref{eq:entropyFinIntegrab} yields  $g(Y)\in L^{1}(\tQ)$, or equivalently $g\in L^{1}(\nu)$, as desired.

On the other hand, the construction of~$\tQ$ is somewhat intricate. It is based on an approximation of~$Q_*$ by a sequence of probability measures $(\tQ_n)$, which simultaneously retain the martingale property and satisfy $h$ is $\tQ_n^1$-a.s.\ uniformly bounded for all $n\in \N$. 
Next, we state a general version of this approximation result, applicable to any measurable function $h:\R\to\R$ and any $Q \in \cMf(\nu)$ satisfying the technical condition~\eqref{eq: technical cond} below. In the proof of Theorem \ref{thm: integrability}, the result will be applied to the specific function~$h$ in $V_*=h(X)(Y-X)+g(Y)$ and $Q = Q_*$.

	\begin{proposition}\label{prop: integrability approx}
	Let $h:\R\to\R$ be measurable and $Q = Q^1 \otimes Q^\bullet \in \cMf(\nu)$.
Suppose that there exists a $Q^1$-integrable function $I : \R \to [0,\infty)$ such that 
\begin{align} \label{eq: technical cond}
	 H \left(Q^\bullet(x')| P^\bullet(x) \right) \leq I(x') \qquad \mbox{for }(Q^1\otimes Q^1)\mbox{-a.a.} \  (x,x').
\end{align}
Then there exist measures $\tQ_n = \tQ_n^1 \otimes \tQ_n^\bullet \in \tcMf(\nu)$ such that 
	\begin{enumerate}
	\item $h$ is $\tQ_n^1$-a.s.\ uniformly bounded for each $n\in \N$,
	\item $\tQ_n\to Q$ in variation,
	\item $H(\tQ_n| P) \to H(Q| P)$.
	\end{enumerate}
	In particular there exists $\tQ \in \tcMf(\nu)$ such that $h$ is uniformly bounded $\tQ^1$-a.s. 
\end{proposition}

The proof is lengthy and deferred to Section \ref{sec: integrability approx}. For ease of reference, we record some standard facts in the next lemma.

\begin{lemma} \label{lemma: measure theory}
	Given probability measures $Q = Q^1 \otimes Q^\bullet$ and $R = R^1 \otimes R^\bullet$ on~$\R^2$,
	
	\begin{enumerate}
	\item $Q\ll R$ if and only if $Q^1 \ll R^1$ and $Q^\bullet \ll R^\bullet$ $Q^1$-a.s.,
	\item  if $Q \ll R$, then %
	$$
		\frac{dQ}{dR}= \frac{dQ^1}{dR^1} \frac{dQ^\bullet}{dR^\bullet} \quad R\as,  %
	$$
	\item 
	$
		H(Q| R) = H(Q^1 | R^1) + E^{Q^1}[ H(Q^\bullet | R^\bullet)].
	$	
	\end{enumerate}
\end{lemma} 

We are now ready to detail the proof of Theorem \ref{thm: integrability}.
\begin{proof}[Proof of Theorem \ref{thm: integrability}]
	Set $\mu:=Q_*^1$. We first construct a function $I$ satisfying~\eqref{eq: technical cond} with~$Q=Q_{*}$.
	By our assumption, there are constants $0 < l < L < \infty$ such that
	\begin{align}\label{eq:bound1}
		l \leq \frac{dP^\bullet(x)}{d\nu}(y) \leq L	\qquad (\mu \otimes\nu)\mbox{-a.a.}\ (x,y).
	\end{align}
	Using Lemma \ref{lemma: measure theory}(i), $Q_* \sim P$ implies that $Q_*^\bullet \sim P^\bullet \sim \nu \ \mu$-a.s. 
	Note also  
	\begin{align}
	\frac{dQ_*^\bullet(x')}{dP^\bullet(x)} (y)&= \frac{dQ_*^\bullet(x')}{dP^\bullet(x')}(y) \, \frac{dP^\bullet(x')}{dP^\bullet(x)}(y)  \nonumber \\
												  &= \frac{dQ_*^\bullet(x')}{dP^\bullet(x')}(y) \, \frac{dP^\bullet(x')}{d\nu}(y) \, \frac{d\nu}{dP^\bullet(x)} (y), \qquad (\mu \otimes \mu \otimes \nu)\mbox{-a.a.}  \ (x,x',y).\label{eq:bound2}
	\end{align}
	Combining \eqref{eq:bound1} and \eqref{eq:bound2} we obtain 
	$$ 
		\log \frac{dQ_*^\bullet(x')}{dP^\bullet(x)}(y) \leq \log \frac{dQ_*^\bullet(x')}{dP^\bullet(x')}(y) + \log (L/l)
	$$
and now integrating against $Q_*^\bullet(x')$ yields
	$$
		H\left(Q_*^\bullet(x')|  P^\bullet(x)\right) \leq H\left(Q_*^\bullet(x') | P^\bullet(x')\right) + \log (L/l) =: I(x'),\qquad (\mu \otimes \mu)\mbox{-a.a.}\ (x,x').
	$$
	Lemma \ref{lemma: measure theory}\,(iii) together with $H(Q_*| P)<\infty$ then implies $I\in L^{1}(\mu)$. 

	Noting that $Q_* \sim\mu \otimes \nu$, Lemma~\ref{le:functionalClosedness} yields that
	$$V_* = h(X)(Y-X) + g(Y)$$
	for some measurable functions $h,g$. Next, we verify that $g$ is $\nu$-integrable with $E^\nu[g] = 0$. Indeed, the function~$I$ satisfies \eqref{eq: technical cond} with $Q = Q_*$, hence Proposition~\ref{prop: integrability approx} provides $\tQ \in \tcMf(\nu)$ such that $h$ is $\tQ^1$-a.s.\ uniformly bounded. This clearly implies $h(X)(Y-X) \in L^1(\tQ)$.
	Recalling from Remark \ref{rmk: psi star int} that $V_*$ is $\tQ$-integrable, we can deduce that $g(Y)\in L^1(\tQ)$; that is, $g \in L^1(\tQ^2)= L^1(\nu)$. We can now conclude from $E^{Q_*}[V_*] = 0$ that $E^\nu[g]= E^{Q_*}[ g(Y)] = 0$, completing the proof that $V_{*}\in \cV_1$.
	
	Recall from Remark~\ref{rmk: psi star int} that $V_* \in L^1(Q)$ for all $Q \in \widetilde{\cM}_{\mathrm{fin}}(\nu)$. Having established $g \in L^1(\nu)$, this implies $h(X)(Y-X) \in L^1(Q)$ and then $E^{Q}[ h(X)(Y-X) ] = 0$ by the martingale property. As $\widetilde{\cM}_{\mathrm{fin}}(\nu)\supset \cMf(\nu)$, this shows that $V_{*}\in \cVadm$.
\end{proof}

\begin{remark}\label{rk:relaxBddnessCond}
As seen in the proof, the boundedness condition in Theorem~\ref{thm: integrability} can be weakened to the following integrability condition:
	\begin{enumerate}
	\item  $P \sim P^1 \otimes \nu$,
	\item there exists a $Q_*^1$-integrable function $I:\R \to [0,\infty)$ such that
		$$
			 E^{Q_*^\bullet(x')} \left[\left| \log \frac{dP^\bullet(x')}{dP^\bullet(x)} \right| \right] \leq I(x') \qquad  \mbox{for } (P^1 \otimes P^1)\mbox{-a.a.}\  (x,x').
		$$
	\end{enumerate}
\end{remark}

\subsection{Proof of Proposition \ref{prop: integrability approx}} \label{sec: integrability approx}

The program for this proof can be sketched as follows. First, we shall identify a sequence $(\mu_n)$ of sub-probability measures $\mu_n \ll Q^1$ such that $h$ is uniformly bounded $\mu_n$-a.e.\ and, when renormalized, $\mu-\mu_n$ dominates~$\mu_n$ in convex order. Strassen's theorem then guarantees the existence of martingale measures $M_n$ with first marginal $\mu_n$ and second marginal $\mu-\mu_n$. The desired measures~$\tQ_n$ have marginals $\mu_n/\mu_n(\R)$ and $\nu$: they will be built by embedding mass $\mu_n(\R)$ according to $\tQ_n^\bullet$ and mass $1-\mu_n(\R)$ according to the composition of $M_n^\bullet$ with $\tQ_n^\bullet$.

Let us first recall that the convex order of two probability measures $\mu,\nu$ can be characterized via their quantile functions $F_\mu^{-1}, F_\nu^{-1}$. Indeed $\mu \preceq_c \nu$ if and only if
\begin{equation} \label{eq: convex order cond}
	\int_u^1 F^{-1}_{\mu}(p) \,dp \leq \int_u^1 F^{-1}_{\nu}(p) \, dp
\end{equation}
for all $u \in [0,1]$, with equality for $u = 0$, see \cite[Theorem~3.A.5]{ShakedShanthikumar.07}.
If $\mu, \nu$ are finite measures with the same total mass, then $\mu \preceq_c \nu$ if and only if $ \mu/\mu(\R) \preceq_c \nu/\nu(\R)$. In particular, we can apply the characterization~\eqref{eq: convex order cond} to these normalized measures. To simplify notation, we omit the normalizing constant and write $F_{\mu}^{-1}$ instead of $F_{\mu/\mu(\R)}^{-1}$ in this case.

As a preparation for the proof of Proposition~\ref{prop: integrability approx}, we first establish two lemmas. Lemma~\ref{lem:convex}\,(i) has the same assertion as \cite[Example~2.4]{BeiglbockJuillet.12} but is obtained with a different, more quantitative argument which is then used in Lemma~\ref{lem:convex}\,(ii) to elaborate on finer properties. Those properties are instrumental for the proof of Lemma~\ref{lem:convex2} which describes a stability property of the convex order that will be applied in the proof of Proposition~\ref{prop: integrability approx}.

\begin{lemma}\label{lem:convex}
Let $A=[a,b]\subseteq \R$. Suppose that $\mu_A$ and $\mu_{B}$ are finite measures with the same mass and zero barycenter such that $\mu_A$ is concentrated on $A$ and $\mu_{B}$ is concentrated on $B:=\R\setminus (a,b)$.
\begin{enumerate}
\item We have $\mu_A\preceq_c \mu_{B}.$ 
\item Define
$$
	E := \left \{ u \in (0,1) : \int_{u}^1 F^{-1}_{\mu_A}(p) \, dp =\int_{u}^1 F^{-1}_{\mu_{B}}(p)\,dp \right\}. 
$$
Then $E$ is of the form $(0, \alpha] \cup [\beta, 1)$ for some $0 \leq \alpha \leq \beta \leq 1$.\footnote{The conventions $(0,0]: = \emptyset$ and $[1,1): = \emptyset$ are used.}
Furthermore,
\begin{enumerate}
\item $E = (0,1)$ if and only if $\mu_A = \mu_{B}$, in which case both measures are concentrated on $\{a, b\}$,

\item  $\mu_{B}((-\infty, a)) = 0$ and $\mu_{B}(a) > \mu_{A}(a) = \alpha$, whenever $\alpha > 0$ and $E \neq (0,1)$,

\item $\mu_{B}((b,\infty)) = 0$ and $\mu_{B}(b) > \mu_{A}(b)  = 1-\beta$, whenever $\beta < 1$ and $E \neq (0,1)$.

\end{enumerate}
\end{enumerate}
\end{lemma}

\begin{proof}
We may assume that $\mu_A$ and $\mu_{B}$ are probability measures. We first show~(i) by verifying~\eqref{eq: convex order cond} for $u \in (0,1)$.
Indeed, define
\begin{align*}
u^{*}:=\sup\left\{p\in (0,1):\  F^{-1}_{\mu_{B}}(p)\le a \right\}.
\end{align*}
Note that $F^{-1}_{\mu_{B}}(p) \geq b$ for all $p \in (u^*, 1)$ and $F^{-1}_{\mu_A}(p) \in [a,b]$ for all $p \in (0,1)$. Hence
\begin{align} \label{eq: cor inequality cond}
\int_u^1 F^{-1}_{\mu_A}(p)\,dp\le \int_u^1 F^{-1}_{\mu_{B}}(p)\,dp
\end{align}
for all $u\in[ u^*,1)$.
Suppose, towards a contradiction, that there exists $ \hat{u} \in (0, u^*)$ such that \eqref{eq: cor inequality cond} holds with the reverse, strict inequality at $u = \hat{u}$.
As $F^{-1}_{\mu_A}(p)\geq a$ for all $p\in (0,1)$ and $F^{-1}_{\mu_{B}}(p)\le a$ for all $p \in (0, u^*)$, we deduce that
\begin{align}\label{eq:4}
\int_0^1 F^{-1}_{\mu_A}(p)\,dp> \int_0^1 F^{-1}_{\mu_{B}}(p)\,dp,
\end{align}
contradicting that $\mu_A$ and $\mu_{B}$ have the same barycenter. This shows (i).

Turning to~(ii), the proof of~(a) is immediate. We can thus assume that there exists $\tilde{u}\in (0,1)$ such that \eqref{eq: cor inequality cond} holds with strict inequality at $u = \tilde{u}$.
If there exists no $\hat{u} \in (0,\tilde{u})$ such that
\begin{align} \label{eq: quantile int eq}
\int_{\hat{u}}^1 F^{-1}_{\mu_A}(p) \,dp  =\int_{\hat{u}}^1 F^{-1}_{\mu_{B}}(p)\,dp,
\end{align}
then $\alpha = 0$. Whereas if such $\hat{u}$ exists, then necessarily $F^{-1}_{\mu_{B}}(p)\le a$ for all $p \in (0, \hat{u}]$, for otherwise $F^{-1}_{\mu_{B}}(\hat{u}) \geq b$ and the equality in~\eqref{eq: quantile int eq} cannot hold.
It follows that
\begin{align}  \label{eq: cor reverse inequality cond}
\int_{u}^1 F^{-1}_{\mu_A}(p) \,dp  \ge \int_{u}^1 F^{-1}_{\mu_{B}}(p)\,dp
\end{align}
for all $u \in (0,\hat{u}]$. Since we have shown the reverse inequality in~(i), we conclude that \eqref{eq: cor reverse inequality cond} holds with equality for all $u \in (0,\hat{u}]$. 
That is, $E$ contains an interval of the form $(0,\alpha]$ for some $\alpha \geq 0$. Changing the integral bounds from $(u,1)$ to $(0,u)$ by subtracting the barycenter on both sides of the above equations, an analogous argument shows that $E$ contains an interval of the form $[\beta,1)$ for some $\beta \in [0,1].$
In conclusion, $E$ is of the form $(0, \alpha] \cup [\beta, 1)$ for $0 \leq \alpha \leq \beta \leq 1$.

To show~(b), suppose that $E \neq (0,1)$ and $\alpha > 0$. The assumption implies that
\begin{align*}
	\int_{0}^\alpha F^{-1}_{\mu_A}(p) \,dp  =\int_{0}^\alpha F^{-1}_{\mu_{B}}(p)\,dp.
\end{align*}
Consequently, $F^{-1}_{\mu_A}(p) = a = F^{-1}_{\mu_{B}}(p)$ for all $p \in (0, \alpha]$. That is, $\mu_{B}((-\infty, a)) = 0$, along with $\mu_A(a) \geq \alpha$ and $\mu_{B}(a) \geq \alpha$. Since
\begin{align*}
	\int_{0}^u F^{-1}_{\mu_A}(p) \,dp  > \int_{0}^u F^{-1}_{\mu_{B}}(p)\,dp \quad  \mbox{for} \quad u \in (\alpha, \beta),
\end{align*}
it is necessary that $F^{-1}_{\mu_A}(p) \in (a,b]$ for  $p \in (\alpha, 1)$ and that $F^{-1}_{\mu_{B}}(p) = a$ for all $p > \alpha$ that are sufficiently close to $\alpha$. We conclude that $\mu_{B}(a) > \mu_{A}(a) = \alpha$, showing~(b). Part~$(c)$ is proved analogously.
\end{proof}

\begin{lemma}\label{lem:convex2}
In the setting of Lemma \ref{lem:convex}, suppose that $\mu_A\neq \mu_{B}$ are probability measures. Let $(\mu^n_A)$, $(\mu_{B}^n)$ be  sequences of probability measures with barycenter zero such that $\mu_A^n\ll \mu_A$ for all $n$ as well as $d_{\mathrm{TV}}(\mu_A^n,\mu_A)\to 0$, $d_{\mathrm{TV}}(\mu_{B}^n,\mu_{B})\to 0$ and $W_1(\mu_{B}^n, \mu_{B}) \to 0$ for $n\to \infty$. Then $\mu_{A}^n\preceq_c \mu_{B}^n$ for all~$n$ sufficiently large.
\end{lemma}

Here $W_{1}$ denotes $1$-Wasserstein distance, and we emphasize that the lemma does not require $\mu_{B}^n \ll \mu_{B}$.

\begin{proof}
Consider the set $E$ in Lemma \ref{lem:convex}\,(ii). As $\mu_A\neq \mu_{B}$, we have $E \neq (0,1)$ and $\alpha < \beta$. Let us consider the cases $\alpha > 0$ and $\alpha = 0$ separately.
If $\alpha > 0$,  Lemma~\ref{lem:convex}\,(ii)\,(b) states that $\bar{\alpha} := \mu_{B}(a) > \mu_{A}(a) = \alpha$. Let $\bar{\alpha}_n := \mu_{B}^n(a)$ and $\alpha_n := \mu_A^n(a)$. Since $d_{\mathrm{TV}}(\mu_A^n,\mu_A)\to 0$ and $d_{\mathrm{TV}}(\mu_{B}^n,\mu_{B})\to 0$ as $n\to \infty$, it follows that $\bar{\alpha}_n \to \bar{\alpha}$ and $\alpha_n \to \alpha$. In view of $\mu_A^n \ll \mu_A$, we conclude that $F_{\mu_A^n}^{-1}(p) \in [a,b]$ for $p \in (0,1)$. 
Fix $\eps < \bar\alpha - \alpha$. Then $\bar{\alpha}_n > \alpha + \eps$ when $n$ is sufficiently large, and 
\begin{equation} \label{eq: left end}
	\int_0^u F^{-1}_{\mu^n_A}(p)\, dp \geq \int_{0}^u F^{-1}_{\mu^n_{B}}(p)\,dp  \quad \mbox{for}\quad u \in (0,\alpha+\eps].
\end{equation}
Whereas in the case $\alpha = 0$, we define $\bar{\alpha} := \mu_{B}((-\infty, a]) > 0$ and $\bar{\alpha}_n := \mu_{B}^n((-\infty, a])$. Again, for a fixed $\eps < \bar\alpha - \alpha=\bar\alpha$, we have $\bar{\alpha}_n > \alpha + \eps$ when $n$ is sufficiently large, and~\eqref{eq: left end} holds.

Similarly, we consider the cases $\beta < 1$ and $\beta = 1$ and define $\bar{\beta}$ accordingly. In either case we can fix $\eps <  \beta - \bar \beta$ and find $n$ sufficiently large so that $\bar{\beta}_n < \beta - \eps$ and
\begin{equation} \label{eq: right end}
	\int_u^1 F^{-1}_{\mu^n_A}(p)\, dp \leq \int_{u}^1 F^{-1}_{\mu^n_{B}}(p)\,dp \quad \mbox{for}\quad u \in [\beta - \eps, 1).
\end{equation}

To complete the proof, it remains to show that the inequality in~\eqref{eq: right end} holds for $u \in O:=(\alpha + \eps, \beta - \eps)$, where $\eps < \min\{ \bar\alpha-\alpha, \beta-\bar\beta\}$ is fixed. Note that $(0,1) \setminus E = (\alpha, \beta)$ and $O \subsetneq (\alpha, \beta)$. As the integrals below are continuous functions of~$u$, there exists $\gamma>0$ such that 
\begin{align}\label{eq: existence of gap}
\int_u^1 F^{-1}_{\mu_A}(p)\,dp +\gamma \le \int_u^1 F^{-1}_{\mu_{B}}(p)\,dp  \quad \mbox{for all}\quad u\in O,
\end{align}
thanks to the definition of $E$. In view of $d_{\mathrm{TV}}(\mu_A^n,\mu_A)\to 0$ and $d_{\mathrm{TV}}(\mu_{B}^n,\mu_{B})\to 0$, the quantile functions converge pointwise. Moreover, we recall that the $1$-Wasserstein distance satisfies 
\begin{equation*}
	W_1(\mu_B^n, \mu_B) = \int_0^1 \left| F^{-1}_{\mu_B^n}(p) - F^{-1}_{\mu_B}(p) \right|  \, dp.
\end{equation*}
Dominated convergence and $W_1(\mu_{B}^n, \mu_{B}) \to 0$ thus imply that
\begin{align*}
\lim_{n\to\infty} \int_u^1 F^{-1}_{\mu_{A}^n}(p)\,dp=\int_u^1 F^{-1}_{\mu_{A}}(p)\,dp, \qquad \lim_{n\to\infty} \int_u^1 F^{-1}_{\mu_{B}^n}(p)\,dp=\int_u^1 F^{-1}_{\mu_{B}}(p)\,dp
\end{align*}
uniformly in $u\in O$.
It now follows from \eqref{eq: existence of gap} that
\begin{align*}
\int_u^1 F^{-1}_{\mu_{A}^n}(p)\,dp&\le \int_u^1 F^{-1}_{\mu_{A}}(p)\,dp+\frac{\gamma}{2}\le \int_u^1 F^{-1}_{\mu_{B}}(p)\,dp-\frac{\gamma}{2}\le \int_u^1 F^{-1}_{\mu_{B}^n}(p)\,dp
\end{align*}
for all $u\in O$ and $n \in \N$ large enough. This completes the proof.
\end{proof}

Given measures $\lambda,\mu$ on $\R$, we write $\lambda \le \mu$ if $\lambda(A) \le \mu(A)$ for all $A\in\cB(\R)$.
The total variation distance between $\lambda$ and $\mu$ is defined as 
	$$
	d_{\mathrm{TV}}(\lambda,\mu) = \sup \left\{ | \lambda(A) - \mu(A)|  : A \in \cB(\R) \right\}.
	$$
If $\lambda \leq \mu$, it is clear that $d_{\mathrm{TV}}(\lambda,\mu) = (\mu - \lambda)(\R)$.
For ease of reference, we record the following consequence.

\begin{lemma}
Let $0\neq \lambda\le \mu$ be finite measures on $\R$. Then the probability measures $\bar \lambda := {\lambda}/{\lambda(\R)}$ and $\bar \mu := {\mu}/{\mu(\R)}$ satisfy $\bar \lambda \ll \bar\mu$ and 
\begin{equation} \label{eq: TV formula}
	d_{\mathrm{TV}}(\bar\lambda,\bar\mu) \leq \frac{(\mu - \lambda)(\R)}{\mu(\R)}.
\end{equation}
\end{lemma}

\begin{proof}
We have $\bar \lambda \ll \bar\mu$ and $\lambda(\R) \leq \mu(\R)$, so that
	\begin{align*}
		d_{\mathrm{TV}}(\bar\lambda,\bar\mu)
		= E^{\bar\mu}\left[\left( 1 - \frac{d\bar\lambda}{d\bar\mu} \right)^+\right]&= \frac{1}{\mu(\R)} E^{\mu} \left[\left( 1 - \frac{\mu(\R)}{\lambda(\R)}\frac{d\lambda}{d\mu} \right)^+ \right]\\
		&\leq \frac{1}{\mu(\R)} E^\mu\left[ \left( 1 - \frac{d\lambda}{d\mu} \right)^+\right]
		= \frac{1}{\mu(\R)} (\mu - \lambda)(\R). \qedhere
	\end{align*}
\end{proof}

We are now ready to give the proof of Proposition \ref{prop: integrability approx}.

\begin{proof}[Proof of Proposition \ref{prop: integrability approx}]
To simplify notation we write $\mu:=Q^1$ and assume without loss of generality that $\mu, \nu$ have zero barycenter. We may also assume that $\mu\neq\delta_0$; otherwise the claim is trivial. The proof proceeds in six steps.\\

\begin{step}\label{st:step1}
Given $\delta > 0$ sufficiently small we shall construct sub-probability measures~$\mu_A, \mu_{B}$ (depending on $\delta$) that satisfy the hypotheses of Lemma~\ref{lem:convex}, and hence $\mu_A\preceq_c \mu_{B}$.

As $\mu$ has zero barycenter and is not a Dirac measure, there exists $c > 0$ such that $\tilde \delta := \mu((-\infty, -c]) \wedge \mu([c, \infty)) > 0$. Then, for all $0 < \delta < \tilde \delta$, there exist an interval $A:= [a,b]$ and a measure $\lambda_A$ such that
\[
 \mu((-\infty, a] \cup [b, \infty))\le \delta, \qquad \mu((a,-c])\ge \delta, \qquad \mu([c, b))\ge \delta
\]
and
\begin{itemize}
	\item  $\lambda_A \le \mu| _A$,
	\item $\lambda_A$ has zero barycenter,
	\item $\mu-\lambda_A$ is concentrated on $\R\setminus (a,b)$ and nonzero.
\end{itemize}
In particular $\mu| _{(a,b)} \leq \lambda_A \leq \mu| _A$.
In fact, if $\mu$ has no atoms, we can set $\lambda_{A}=\mu| _{A}$. In the presence of atoms at $a$ or $b$, we may have to remove part of that mass so that $\lambda_A$ has zero barycenter and $\lambda_A\neq \mu$.
Define
$$\mu_A:=\left(\frac{1}{\lambda_A(\R)}-1\right)\lambda_A \quad \text{and}\quad \mu_{B}:=\mu-\lambda_A.$$
Evidently, the hypotheses of Lemma \ref{lem:convex} are satisfied, and hence $\mu_A\preceq_c \mu_{B}$.
\end{step}

\begin{step}\label{st:step2}
Fix $\eps>0$. As $h$ takes values in $\R$ there exists a set $A^\eps\subseteq A$ such that $h$ is uniformly bounded on $A^\eps$ and 
\begin{align}\label{eq:rhs}
\mu(A\setminus A^\eps)\le \frac{\delta \wedge \eps}{2} \left( 1 \wedge \frac{c}{| a| \vee | b| }\right).
\end{align}
The choice of the upper bound in \eqref{eq:rhs} ensures that
\begin{align*}
\mu\left((a,-c]\cap A^\eps\right)\ge \frac{\delta}{2},\quad \mu\left([c, b)\cap A^\eps\right)\ge \frac{\delta}{2}
\end{align*}
and, writing $X$ for the identity function on $\R$ in an abuse of notation,
\begin{align*}
\left| E^{\lambda_A} \left[ \1_{A^\eps} X \right]\right| &\le \left| E^{\lambda_A}\left[\1_A X \right] \right| + \left| E^{\lambda_A}\left[\1_{A\setminus A^\eps} X \right] \right| \\
&\le 0+ (| a| \vee | b| )  \frac{(\delta\wedge\eps) c}{2 (| a| \vee | b| )}=\frac{(\delta\wedge \eps) c}{2}.
\end{align*}
By restricting $\lambda_A$ to $A^\eps$ and possibly removing some mass on $(a,-c]$ or $[c,b)$, we can construct a measure $\lambda^\ep_A$ that is concentrated on $A^\eps$, satisfies $\lambda_A^\eps\leq \lambda_A \leq \mu$, has zero barycenter and
$$
	d_{\mathrm{TV}}(\lambda^\eps_A,\lambda_A)=(\lambda_A - \lambda^\ep_A)(\R) \leq \mu(A \setminus A^\eps) + \frac{1}{c} \left| E^{\lambda_A}[\1_{A^\eps} X ]\right|   \leq \delta \wedge \eps.
$$
In consequence,
\begin{equation} \label{eq: lambda ep size}
	\lambda^\ep_A (\R) = \lambda_A(\R) - (\lambda_A - \lambda^\ep_A)(\R)  \geq 1 - \delta - \delta \wedge \eps.
\end{equation}
\end{step}

\begin{step}\label{st:step3}
Let $\delta \leq \tilde\delta$ be fixed and consider a sequence $(\eps_n)_{n\in \N} \downarrow 0$.  For each $n\in\N$ we apply Step~\ref{st:step2} to obtain
\begin{equation} \label{eq: def mu n}
	\mu^n_A := \left(\frac{1}{\lambda^n_A(\R)}-1\right) \lambda^n_A \quad \text{and} \quad
	\mu^n_{B}:=\mu - \lambda^n_A,
\end{equation}
where $\lambda^n_A := \lambda^{\ep_n}_A$.
Both measures have the same total mass and zero barycenter for all $n\in \N$. Moreover we have $\lambda^n_A \leq \lambda_A$ and $d_{\mathrm{TV}}(\lambda_A,\lambda_A^n)= (\lambda_A - \lambda^n_A)(\R) \leq \eps_n \downarrow 0$.

In order to apply Lemma~\ref{lem:convex2}, we first need to scale the measures $\mu_A, \mu_{B}, \mu_{A}^n$ and $\mu_{B}^n$ so that they are probability measures. Set
\begin{equation} \label{eq: def mu bar}
	\bar{\mu}_A = \frac{\lambda_A}{\lambda_A(\R)} \quad \text{and} \quad \bar{\mu}_{B} = \frac{\mu - \lambda_A}{1 - \lambda_{A}(\R)}
\end{equation}
and define $\bar{\mu}_A^n$ and $\bar{\mu}^n_{B}$ analogously. Observe that $\bar{\mu}^n_A \ll \bar{\mu}_A$ and 
$$
	d_{\mathrm{TV}}(\bar{\mu}_A^n, \bar{\mu}_A) \leq \frac{(\lambda_A - \lambda_A^n)(\R)}{\lambda_A(\R)} \leq \frac{\ep_n}{\lambda_A(\R)} \downarrow 0
$$
by \eqref{eq: TV formula}. Similarly we have $d_{\mathrm{TV}}(\bar{\mu}_{B}^n, \bar{\mu}_{B}) \to 0$. In particular it suffices to show that $
	E^{\mu^n_{B}}[| X|]\to E^{\mu_B} [| X|]
$
in order to verify $W_1( \bar \mu^n_{B}, \bar \mu_{B}) \to 0$.
In light of the definition of $\bar \mu_{B}$ in \eqref{eq: def mu bar} and $d_{\mathrm{TV}}(\lambda_{A}^n, \lambda_{A}) \to 0$  this readily follows from $E^{\lambda_A^n} [| X|]  \to E^{\lambda_A}[ | X| ]$.

Now we are in a position to apply Lemma \ref{lem:convex2} to $\bar\mu_A, \bar\mu_{B}, \bar\mu_{A}^n$ and $\bar\mu_{B}^n$, which yields $n_0 \in \N$ such that $\bar\mu_{A}^{n_0} \preceq_c \bar\mu^{n_0}_{B}$.
Since the convex order is invariant under scaling, it follows that
\begin{equation} \label{eq: mu star}
	\mu^*_{A}:=\mu_{A}^{n_0} \preceq_c  \mu^{n_0}_{B}=:\mu^*_{B}.
\end{equation}
We recall that $h$ is uniformly bounded on $A^* := A^{\eps_{n_0}}$ by construction and define 
\begin{equation} \label{eq: lambda A star}
	\lambda_A^*:= \lambda_A^{n_0}
\end{equation}
in preparation for Step~\ref{st:step5} below.
\end{step}

\begin{step}\label{st:step4} %
By \cite[Theorem~8]{Strassen.65} the relation $\mu_A^*\preceq_c \mu_{B}^*$ implies the existence of a mean-preserving probability kernel $M_\delta^\bullet:\R\to \mathcal{P}(\R)$ sending~$\mu_A^*$ to~$\mu_{B}^*$; that is, $M_\delta := \mu_A^* \otimes M_\delta^\bullet \in \Pi(\mu_A^*, \mu_{B}^*)$ and $M_\delta^\bullet(x)$ has barycenter~$x$ for all~$x\in\R$.

Recall that  $Q= \mu \otimes Q^\bullet\in\mathcal{M}_{\mathrm{fin}}(\nu)$ and denote by $Q^\bullet_{\delta}$ the composition of~$M_\delta^\bullet$ with~$Q^\bullet$,
\begin{align} \label{eq: kappa delta}
	Q^\bullet_{\delta}(x,C) := E^{ M_\delta^\bullet(x)}[Q^\bullet(\cdot ,C)]\quad \text{for} \quad C\in \mathcal{B}(\R).
\end{align}
Note that $Q^\bullet_{\delta}$ is again mean-preserving. %
\end{step}

\begin{step}\label{st:step5}
For $\delta \leq \tilde\delta$, we set
\begin{align}\label{eq:QdeltaDef}
\tQ_\delta := \lambda_{A}^* \otimes Q^\bullet + \mu_{A}^* \otimes Q^\bullet_\delta,
\end{align}
where $\mu_A^*, \lambda_{A}^*$ and $Q^\bullet_\delta$ were defined in  \eqref{eq: mu star}, \eqref{eq: lambda A star} and \eqref{eq: kappa delta}, respectively (the set~$A$ and the measures $\lambda_{A}^*,\mu_A^*$ depend on $\delta$).
In view of~\eqref{eq: def mu n}, the first marginal of $\tQ_\delta$ is the probability measure
$
	\bar \mu^*_{A} := {\lambda_{A}^*}/{\lambda_{A}^*(\R)}.
$

We claim that $\tQ_\delta \in \widetilde{\cM}(\nu)$. 
Indeed, recall $\mu \otimes Q^\bullet = Q \sim P$ and observe that $\mu_{A}^* \sim \lambda_A^* \ll \mu$, $Q^\bullet \sim \nu$ $\mu$-a.s. and $Q^\bullet_\delta \ll \nu$ $\mu_A^*$-a.s. 
Lemma~\ref{lemma: measure theory}\,(i) then yields $\tQ_\delta \ll P$. Moreover, the martingale property $E^{\tQ_\delta}[Y| X] = X$ follows from the fact that $Q^\bullet$ and $Q^\bullet_\delta$ are mean-preserving. Finally,
recall that $Q^\bullet_\delta$ is the composition of $M_\delta^\bullet$ with $Q^\bullet$, and $\mu_A^* \otimes M_\delta^\bullet \in \Pi(\mu_A^*, \mu_{B}^*)$. Thus $\mu = \lambda_A^* + \mu_{B}^*$ and $\mu \otimes Q^\bullet \in \Pi(\mu, \nu)$ imply $\tQ_\delta \in \Pi \left( \bar \mu^*_{A}, \nu \right)$ and in particular $\tQ_\delta^{2}=\nu$. In summary, $\tQ_\delta \in \widetilde{\cM}(\nu)$ as desired.

From Step~2 and \eqref{eq: lambda ep size} we see that $\lambda_A^* \leq \mu$ and $\lambda_A^*(\R) \geq 1 - 2 \delta$, hence $\lambda_A^*\to \mu$ and $\mu_{A}^*\to0$ in variation as $\delta\to0$. It is now clear from the definition~\eqref{eq:QdeltaDef} that $\tQ_\delta\to \mu \otimes Q^\bullet=Q$ in variation. Moreover, as $\lambda_A^*$ is concentrated on $A^*$ (cf.\ Step~\ref{st:step3}) and $h$ is uniformly bounded on $A^*$, we see that $h$ is $\bar \mu^*_{A}$-a.s.\ uniformly bounded for every $\delta \leq \tilde\delta$. This proves Proposition~\ref{prop: integrability approx}\,(i),(ii) after choosing $\delta=\delta(n)$ small enough, modulo showing that $H(\tQ_\delta|P)<\infty$ for small~$\delta$ (which will follow from the next step).
\end{step}

\begin{step}\label{st:step6}
As $\liminf_{\delta\to 0} H(\tQ_\delta| P)  \ge H(Q| P)$ due to the lower semicontinuity of~$H(\cdot |P)$ and the convergence $\tQ_\delta\to Q$, it remains to show 
\begin{equation}\label{eq:step6limsup}
  \limsup_{\delta\to 0} H(\tQ_\delta |  P)\le H(Q |  P).
\end{equation}
 Note that $\tQ_\delta$ is a convex combination of two probability measures:
\begin{align*}
	\tQ_\delta= \lambda_A^* (\R)  \bar \mu^*_{A} \otimes Q^\bullet + \left[ 1-\lambda_A^* (\R) \right]  \bar \mu^*_{A} \otimes Q^\bullet_\delta.
\end{align*}
As $H(\cdot |P)$ is convex, it follows that
\begin{align}\label{eq:step6convex}
	H(\tQ_\delta |  P) \le \lambda_A^*(\R)  H\left( \bar \mu^*_{A} \otimes Q^\bullet |  P\right) +  \left[ 1-\lambda_A^*(\R) \right] H\left( \bar \mu^*_{A} \otimes Q^\bullet_\delta|  P\right).
\end{align}
We show that the first term converges to $H(Q |  P)$ and the second converges to zero.
Indeed, Lemma~\ref{lemma: measure theory}\,(iii) yields 
$
	H\left( \bar \mu^*_{A} \otimes Q^\bullet|  P \right) = H\left(\bar \mu^*_{A} |  P^1\right)
	+ E^{\bar \mu^*_{A}}[ H(Q^\bullet | P^\bullet)]
$, where
\begin{align}\label{eq:step6firstEnt}
	H\left(\bar \mu^*_{A} |  P^1 \right) &=\frac{1}{\lambda_A^*(\R)} E^{\lambda_A^*}\left[ \log\left(\frac{d \lambda_A^*}{dP^1} \right) \right] - \log \lambda_A^*(\R) \nonumber \\
												&\to E^\mu\left[ \log\left(\frac{d \mu}{dP^1} \right) \right] =H(\mu |  P^1)
\end{align}
by dominated convergence and $\lambda_A^*(\R) \ge 1-2\delta$. Similarly, $E^{\bar \mu^*_{A}}[ H(Q^\bullet | P^\bullet)] \to E^{\mu}[ H(Q^\bullet | P^\bullet)]$, so that the first term in~\eqref{eq:step6convex} satisfies
$$
  \lambda_A^*(\R)  H\left( \bar \mu^*_{A} \otimes Q^\bullet |  P\right) \to H(\mu |  P^1) + E^{\mu}[ H(Q^\bullet | P^\bullet)] = H(Q|R). 
$$
It remains to show that the second term in~\eqref{eq:step6convex} converges to zero,
$$ 
	[1-\lambda_A^*(\R)]  H\left( \bar \mu^*_{A} \otimes Q^\bullet_\delta|  P\right) \to 0.
$$ 
Using again Lemma \ref{lemma: measure theory}\,(iii),
\begin{align}\label{eq:step6decomp}
	H\left( \bar \mu^*_{A} \otimes Q^\bullet_\delta|  P \right) = H\left(\bar \mu^*_{A} |  P^1\right) + E^{\bar{\mu}_A^*}[ H( Q^\bullet_{\delta} |  P^\bullet)].
\end{align}
In view of~\eqref{eq:step6firstEnt} and $\lambda_A^*(\R)\to1$,  it follows that $\left[ 1-\lambda_A^*(\R) \right] H\left(\bar \mu^*_{A} |  P^1\right)\to0$.
For the second term in~\eqref{eq:step6decomp}, we use the definitions of $\bar{\mu}_A^*$ and $\mu_A^*$ to see that
\begin{align*}
[1-\lambda_A^*(\R)] E^{\bar{\mu}_A^*} \left[H( Q^\bullet_{\delta}|  P^\bullet)\right]
=\frac{1-\lambda_A^*(\R)}{\lambda_A^*(\R)} E^{\lambda_A^*} \left[H( Q^\bullet_{\delta}|  P^\bullet)\right]
= E^{\mu_A^*}[ H( Q^\bullet_{\delta}|  P^\bullet)].
\end{align*}
In view of \eqref{eq: kappa delta}, Jensen's inequality and convexity of $H$ imply
\begin{align*}
E^{\mu_A^*}[ H( Q^\bullet_{\delta}|  P^\bullet)]
= E^{\mu_A^*} \left[H\left( E^{ M_\delta^\bullet(X)}[Q^\bullet]  |  P^\bullet(X) \right) \right]
\le E^{\mu_A^*}\left[E^{M_\delta^\bullet(X)}[ H( Q^\bullet |  P^\bullet(X)) ]\right].
\end{align*}
Lastly, the assumptions that $H(Q^\bullet(x')| P^\bullet(x)) \leq I(x')$ for $(\mu\otimes \mu)$-a.a $(x,x')$ and $I \in L^{1}(\mu)$ together with the facts that $\mu_A^* \otimes M_\delta^\bullet \in \Pi(\mu_A^*, \mu_{B}^*)$ and $\mu_B^* \to 0$ yield
\begin{align*}
 E^{\mu_A^*}\left[E^{M_\delta^\bullet(X)}[ H( Q^\bullet |  P^\bullet(X)) ]\right]
\le E^{\mu_A^*}\left[E^{M_\delta^\bullet(X)}[ I]\right] 
=  E^{\mu_{B}^*}[ I ]  \to 0
\end{align*}
by the dominated convergence theorem.
This shows~\eqref{eq:step6limsup} and hence Proposition~\ref{prop: integrability approx}\,(iii), completing the proof. \qedhere
\end{step}
\end{proof}

\bibliography{stochfin_LZ} 

\newcommand{\dummy}[1]{}
\begin{thebibliography}{10}

\bibitem{AcciaioLarssonSchachermayer.17}
B.~Acciaio, M.~Larsson, and W.~Schachermayer.
\newblock The space of outcomes of semi-static trading strategies need not be
  closed.
\newblock {\em Finance Stoch.}, 21(3):741--751, 2017.

\bibitem{AltschulerNilesWeedStromme.21}
J.~M. Altschuler, J.~Niles-Weed, and A.~J. Stromme.
\newblock Asymptotics for semi-discrete entropic optimal transport.
\newblock {\em Preprint arXiv:2106.11862v1}, 2021.

\bibitem{Avellaneda.98}
M.~Avellaneda.
\newblock Minimum-relative-entropy calibration of asset-pricing models.
\newblock {\em Int. J. Theor. Appl. Finance}, 1(4):447--472, 1998.

\bibitem{AvellanedaEtAl.01}
M.~Avellaneda, R.~Buff, C.~Friedman, N.~Grandechamp, L.~Kruk, and J.~Newman.
\newblock Weighted {M}onte {C}arlo: a new technique for calibrating
  asset-pricing models.
\newblock {\em Int. J. Theor. Appl. Finance}, 4(1):91--119, 2001.

\bibitem{BeiglbockHenryLaborderePenkner.11}
M.~Beiglb{\"o}ck, P.~Henry-Labord{\`e}re, and F.~Penkner.
\newblock Model-independent bounds for option prices: a mass transport
  approach.
\newblock {\em Finance Stoch.}, 17(3):477--501, 2013.

\bibitem{BeiglbockJuillet.12}
M.~Beiglb{\"o}ck and N.~Juillet.
\newblock On a problem of optimal transport under marginal martingale
  constraints.
\newblock {\em Ann. Probab.}, 44(1):42--106, 2016.

\bibitem{Beurling.60}
A.~Beurling.
\newblock An automorphism of product measures.
\newblock {\em Ann. of Math. (2)}, 72:189--200, 1960.

\bibitem{BorweinLewis.92}
J.~M. Borwein and A.~S. Lewis.
\newblock Decomposition of multivariate functions.
\newblock {\em Canad. J. Math.}, 44(3):463--482, 1992.

\bibitem{BreedenLitzenberger.78}
D.~T. Breeden and R.~H. Litzenberger.
\newblock Prices of state-contingent claims implicit in option prices.
\newblock {\em J. Bus.}, 51(4):621--651, 1978.

\bibitem{Csiszar.75}
I.~Csisz\'{a}r.
\newblock {$I$}-divergence geometry of probability distributions and
  minimization problems.
\newblock {\em Ann. Probability}, 3:146--158, 1975.

\bibitem{Cuturi.13}
M.~Cuturi.
\newblock Sinkhorn distances: Lightspeed computation of optimal transport.
\newblock In {\em Advances in Neural Information Processing Systems 26}, pages
  2292--2300. 2013.

\bibitem{DalangMortonWillinger.90}
R.~C. Dalang, A.~Morton, and W.~Willinger.
\newblock Equivalent martingale measures and no-arbitrage in stochastic
  securities market models.
\newblock {\em Stochastics and Stochastic Rep.}, 29(2):185--201, 1990.

\bibitem{DelbaenEtAl.02}
F.~Delbaen, P.~Grandits, T.~Rheinl{\"a}nder, D.~Samperi, M.~Schweizer, and
  C.~Stricker.
\newblock Exponential hedging and entropic penalties.
\newblock {\em Math. Finance.}, 12:99--123, 2002.

\bibitem{DeligiannidisDeBortoliDoucet.21}
G.~Deligiannidis, V.~{De Bortoli}, and A.~Doucet.
\newblock Quantitative uniform stability of the iterative proportional fitting
  procedure.
\newblock {\em Preprint arXiv:2108.08129v1}, 2021.

\bibitem{Follmer.88}
H.~F\"{o}llmer.
\newblock Random fields and diffusion processes.
\newblock In {\em \'{E}cole d'\'{E}t\'{e} de {P}robabilit\'{e}s de
  {S}aint-{F}lour {XV}--{XVII}, 1985--87}, volume 1362 of {\em Lecture Notes in
  Math.}, pages 101--203. Springer, Berlin, 1988.

\bibitem{FollmerGantert.97}
H.~F\"{o}llmer and N.~Gantert.
\newblock Entropy minimization and {S}chr\"{o}dinger processes in infinite
  dimensions.
\newblock {\em Ann. Probab.}, 25(2):901--926, 1997.

\bibitem{FollmerSchied.11}
H.~F{\"{o}}llmer and A.~Schied.
\newblock {\em Stochastic Finance: An Introduction in Discrete Time}.
\newblock W. de Gruyter, Berlin, 3rd edition, 2011.

\bibitem{Frittelli.00}
M.~Frittelli.
\newblock The minimal entropy martingale measure and the valuation problem in
  incomplete markets.
\newblock {\em Math. Finance}, 10(1):39--52, 2000.

\bibitem{GalichonHenryLabordereTouzi.11}
A.~Galichon, P.~Henry-Labord{\`e}re, and N.~Touzi.
\newblock A stochastic control approach to no-arbitrage bounds given marginals,
  with an application to lookback options.
\newblock {\em Ann. Appl. Probab.}, 24(1):312--336, 2014.

\bibitem{GigliTamanini.21}
N.~Gigli and L.~Tamanini.
\newblock Second order differentiation formula on {${\mathrm{RCD}}^{*}(K,N)$}
  spaces.
\newblock {\em J. Eur. Math. Soc. (JEMS)}, 23(5):1727--1795, 2021.

\bibitem{GuoObloj.19}
G.~Guo and J.~Ob{\l}{\'o}j.
\newblock Computational methods for martingale optimal transport problems.
\newblock {\em Ann. Appl. Probab.}, 29(6):3311--3347, 2019.

\bibitem{Guyon.20}
J.~Guyon.
\newblock The joint {S\&P} 500/{VIX} smile calibration puzzle solved.
\newblock {\em Risk}, 2020.

\bibitem{Guyon.21}
J.~Guyon.
\newblock Dispersion-constrained martingale {S}chr{\"o}dinger problems and the
  exact joint {S\&P} 500/{VIX} smile calibration puzzle.
\newblock {\em Preprint SSRN:3853237}, 2021.

\bibitem{HenryLabordere.19}
P.~Henry-Labord{\`e}re.
\newblock From (martingale) {S}chr{\"o}dinger bridges to a new class of
  stochastic volatility model.
\newblock {\em Preprint SSRN:3353270}, 2019.

\bibitem{Hobson.98}
D.~Hobson.
\newblock Robust hedging of the lookback option.
\newblock {\em Finance Stoch.}, 2(4):329--347, 1998.

\bibitem{JacodShiryaev.98}
J.~Jacod and A.~N. Shiryaev.
\newblock Local martingales and the fundamental asset pricing theorems in the
  discrete-time case.
\newblock {\em Finance Stoch.}, 2(3):259--273, 1998.

\bibitem{Leonard.14}
C.~L\'{e}onard.
\newblock A survey of the {S}chr\"{o}dinger problem and some of its connections
  with optimal transport.
\newblock {\em Discrete Contin. Dyn. Syst.}, 34(4):1533--1574, 2014.

\bibitem{DeMarchHenryLabordere.19}
H.~D. March and P.~Henry-Labord{\`e}re.
\newblock Building arbitrage-free implied volatility: Sinkhorn's algorithm and
  variants.
\newblock {\em Preprint SSRN:3326486}, 2019.

\bibitem{Nutz.20}
M.~Nutz.
\newblock {\em Introduction to Entropic Optimal Transport}.
\newblock Lecture notes, Columbia University, 2021.
\newblock
  \url{https://www.math.columbia.edu/~mnutz/docs/EOT_lecture_notes.pdf}.

\bibitem{NutzWiesel.21}
M.~Nutz and J.~Wiesel.
\newblock Entropic optimal transport: Convergence of potentials.
\newblock {\em Probab. Theory Related Fields, to appear}.
\newblock arXiv:2104.11720v2.

\bibitem{NutzWiesel.22}
M.~Nutz and J.~Wiesel.
\newblock Stability of {S}chr\"{o}dinger potentials and convergence of
  {S}inkhorn's algorithm.
\newblock {\em Preprint arXiv:2201.10059v1}, 2022.

\bibitem{NutzWieselZhao.22a}
M.~Nutz, J.~Wiesel, and L.~Zhao.
\newblock Limits of semistatic trading strategies.
\newblock {\em Preprint}, 2022.

\bibitem{PeyreCuturi.19}
G.~Peyr{\'e} and M.~Cuturi.
\newblock Computational optimal transport: With applications to data science.
\newblock {\em Foundations and Trends in Machine Learning}, 11(5-6):355--607,
  2019.

\bibitem{Rogers.94}
L.~C.~G. Rogers.
\newblock Equivalent martingale measures and no-arbitrage.
\newblock {\em Stochastics Stochastics Rep.}, 51(1-2):41--49, 1994.

\bibitem{RuschendorfThomsen.93}
L.~R\"{u}schendorf and W.~Thomsen.
\newblock Note on the {S}chr\"{o}dinger equation and {$I$}-projections.
\newblock {\em Statist. Probab. Lett.}, 17(5):369--375, 1993.

\bibitem{RuschendorfThomsen.97}
L.~R\"{u}schendorf and W.~Thomsen.
\newblock Closedness of sum spaces and the generalized {S}chr\"{o}dinger
  problem.
\newblock {\em Teor. Veroyatnost. i Primenen.}, 42(3):576--590, 1997.

\bibitem{Schachermayer.01}
W.~Schachermayer.
\newblock Optimal investment in incomplete markets when wealth may become
  negative.
\newblock {\em Ann. Appl. Probab.}, 11(3):694--734, 2001.

\bibitem{ShakedShanthikumar.07}
M.~Shaked and J.~G. Shanthikumar.
\newblock {\em Stochastic orders}.
\newblock Springer Series in Statistics. Springer, New York, 2007.

\bibitem{Strassen.65}
V.~Strassen.
\newblock The existence of probability measures with given marginals.
\newblock {\em Ann. Math. Statist.}, 36:423--439, 1965.

\bibitem{Zariphopoulou.01}
T.~Zariphopoulou.
\newblock A solution approach to valuation with unhedgeable risks.
\newblock {\em Finance Stoch.}, 5(1):61--82, 2001.

\end{thebibliography}
\bibliographystyle{abbrv}

\end{document}